\documentclass[12pt,a4paper]{article}
\usepackage[latin1]{inputenc}
\usepackage{amsmath,amssymb,srcltx}
\usepackage{amsmath,amssymb,amsthm,amstext} 
\usepackage{stmaryrd}
\usepackage{bbm}
\usepackage{yfonts}
\usepackage{tikz}

\newcommand{\bt}{\begin{thr}{\bf Theorem. }} 
\newcommand{\satz}{\begin{thr}{\bf Theorem. }\rm} 
\newcommand\E{\mathbb{E}}

\newcommand\p{\mathbb{P}}

\newcommand\F{\mathcal{F}}

\renewcommand\i{\infty}
\newcommand\al{\alpha}

\newcommand\la{\lambda}

\newcommand\Vliq{V^{\text{\it{liq}}}}

\newtheorem{theorem}{Theorem}[section]  

\newtheorem{definition}[theorem]{Definition}
\newtheorem{lemma}[theorem]{Lemma}
\newtheorem{proposition}[theorem]{Proposition}

\theoremstyle{definition}
\newtheorem*{acka}{Acknowledgement}

\newtheorem{remark}[theorem]{Remark}

\begin{document}

\title{Admissible Trading Strategies under Transaction Costs}

\author{Walter Schachermayer\footnote{Fakult\"at f\"ur Mathematik, Universit\"at Wien, Nordbergstrasse 15, A-1090 Wien, {\tt walter.schachermayer@univie.ac.at}. Partially supported by the Austrian Science Fund (FWF) under grant P25815, the European Research Council (ERC) under grant FA506041 and by the Vienna Science and Technology Fund (WWTF) under grant MA09-003.}}

\date{\today}
\maketitle

\begin{abstract}
A well known result in stochastic analysis reads as follows: for an $\mathbb{R}$-valued super-martingale $X = (X_t)_{0\leq t \leq T}$ such that the terminal value $X_T$ is non-negative, we have that the entire process $X$ is non-negative. An analogous result holds true in the  no arbitrage theory of mathematical finance: under the assumption of no arbitrage, an admissible portfolio process $x+(H\cdot S)$ verifying $x+(H\cdot S)_T\geq 0$ also satisfies $x+(H\cdot S)_t\geq 0,$ for all $0 \leq t \leq T$.

In the present paper we derive an analogous result in the presence of transaction costs. In fact, we give two versions: one with a numéraire-based, and one with a numéraire-free notion of admissibility. It turns out that this distinction on the primal side perfectly corresponds to the difference between local martingales and true martingales on the dual side. 

A counter-example reveals that the consideration of transaction costs makes things more delicate than in the frictionless setting.
\end{abstract}

\section{A Theorem on Admissibility}

We consider a stock price process $S=(S_t)_{0\le t\le T}$ in continuous time with a fixed horizon $T.$ This stochastic process is assumed to be based on a filtered probability space
$(\Omega, \F, (\F_t)_{0\le t\le T}, \p),$ satisfying the usual conditions of completeness and right continuity. We assume that $S$ is adapted and has \emph{c\`adl\`ag} (right continuous, left limits), 
and strictly positive trajectories, i.e.~the function $t\to S_t (\omega)$ is c\`adl\`ag and strictly positive, for almost each $\omega \in\Omega.$ 

In mathematical finance a key assumption is that the process $S$ is {\it free of arbitrage}. The Fundamental Theorem of Asset Pricing states that this property is {\it essentially} equivalent to the property that $S$ admits an equivalent local martingale measure (see, \cite{HK79}, \cite{Delbaen1994}, or the books \cite{DelbaenS},\cite{KSh98}).

\begin{definition}\label{def1.1}
The process $S$ admits an equivalent local martingale measure, if there is a probability measure $Q \sim \mathbb{P}$ such that $S$ is a local martingale under $Q$.
\end{definition}

Fix a process $S$ satisfying the above assumption and note that Def.1.1 implies in particular that $S$ is a semi-martingale as this property is invariant under equivalent changes of measure. Turning to the theme of the paper, we now consider {\it trading strategies}, $\mbox{i.e.}$ $S$-integrable predictable processes $H=(H_t)_{0 \leq t \leq T}$. We call H {\it admissible} if there is $M>0$ such that
\begin{equation}\label{1}
(H \cdot S)_t \geq -M, \qquad \mathbb{P}-a.s.  \quad \mbox{for} \qquad 0 \leq t \leq T.
\end{equation}
The stochastic integral
\begin{equation}
(H \cdot S)_t = \int^t_0 H_u dS_u, \qquad\qquad 0 \leq t \leq T,
\end{equation}
then is a local $Q$-martingale by a result of Ansel-Stricker under each equivalent local martingale measure $Q$ (see \cite{AS94} and \cite{St02}). Assumption \eqref{1} also implies that the local martingale $H \cdot S$ is a {\it super-martingale} (see \cite{DelbaenS}, Prop.7.2.7) under each equivalent local martingale measure $Q$. We thus infer from the easy result mentioned in the first line of the abstract that $(H \cdot S)_T \geq -x$ almost surely implies that $(H \cdot S)_t \geq -x$ almost surely under $Q$ (and therefore also under $\mathbb{P}$), for all $0 \leq t \leq T$. In fact, we may replace the deterministic time $t$ by a $[0,T]$-valued stopping time $\tau.$
\vskip10pt
We resume our findings in the subsequent well-known Proposition (compare \cite{63}, Prop.4.1).

\begin{proposition}\label{pro4.12}
Let the process $S$ admit an equivalent local martingale measure, let $H$ be admissible, and suppose that there is $x\in\mathbb{R}_+$ such that

\begin{equation}
x+(H \cdot S)_T \geq0, \qquad\qquad \mathbb{P}-a.s.
\end{equation}
Then
\begin{equation}
x+(H \cdot S)_{\tau} \geq0, \qquad\qquad \mathbb{P}-a.s.
\end{equation}
for every $[0,T]$-valued stopping time $\tau$.
\end{proposition}
\vskip6pt

We now introduce transaction costs: fix $0 \leq \lambda < 1$. We define the {\it bid-ask spread} as the interval $[(1-\lambda)S, S]$. The interpretation is that an agent can buy the stock at price $S$, but sell it only at price $(1-\lambda) S$. Of course, the case $\lambda=0$ corresponds to the usual frictionless theory. 
\newline
\vskip6pt
In the setting of transaction costs the notion of {\it consistent price systems}, which goes back to \cite{JouiKall} and \cite{CvitKara}, plays a role analogous to the notion of equivalent martingale measures in the frictionless theory (Definition \ref{def1.1}).

\begin{definition}\label{def4.1}
Fix $1>\la \geq0.$ A process $S=(S_t)_{0\le t\le T}$ satisfies the condition $(CPS^\la)$ of \textnormal{having a consistent price system under transaction costs $\la$} 
if there is a process $\widetilde{S}=(\widetilde{S}_t)_{0\le t\le T},$ such that
\begin{equation}\label{p3}
(1-\la)S_t \le \widetilde{S}_t \le S_t, \qquad\qquad 0\le t\le T,
\end{equation}
as well as a probability measure $Q$ on $\F$, equivalent to $\p,$ such that $(\widetilde{S}_t)_{0\le t\le T}$ is a local martingale under $Q$.

We say that $S$ \textnormal{admits consistent price systems for arbitrarily small transaction costs} if $(CPS^\la)$ is satisfied, for all $1 > \la >0.$
\end{definition}

For continuous process $S$, in \cite{GRS08} the condition of {\it admitting consistent price systems} for arbitrarily small transaction costs has been related to the condition of {\it no arbitrage} 
under arbitrarily small transaction costs, thus proving a version of the Fundamental Theorem of Asset Pricing under small transaction costs (compare \cite{KS02} for a large amount of related material).
  
It is important to note that we {\it do not assume} that $S$ is a semi-martingale as one is forced to do in the frictionless theory \cite[Theorem 7.2]{Delbaen1994}.
Only the process $\widetilde{S}$ appearing in Definition \ref{def4.1} has to be a semi-martingale, as it becomes a local martingale after passing to an equivalent measure $Q.$

\vskip10pt
\indent
To formulate a result analogous to Proposition \ref{pro4.12} in the setting of transaction costs we have to define the notion of $\mathbb{R}^2$-valued {\it self-financing trading strategies}.

\begin{definition}\label{def4.2}
Fix a strictly positive stock price process $S=(S_t)_{0\le t\le T}$ with c\`adl\`ag paths, as well as transaction costs $1>\la >0.$

A \textnormal{self-financing trading strategy starting with zero endowment} is a pair of predictable, finite variation processes $(\varphi^0_t,\varphi^1_t)_{0\le t\le T}$ 
such that\\
\vskip6pt
\indent
$(i)$ \ $\varphi^0_0 =\varphi^1_0=0$, \\
\vskip3pt
$(ii)$ denoting by $\varphi^0_t=\varphi_t^{0,\uparrow} -\varphi_t^{0,\downarrow} $ and $\varphi^1_t=\varphi_t^{1,\uparrow} -\varphi_t^{1,\downarrow},$ the canonical decompositions of $\varphi^0$ and $\varphi^1$ into 
the difference of two increasing processes, starting at $\varphi_0^{0,\uparrow} =\varphi_0^{0,\downarrow} = \varphi_0^{1,\uparrow} =\varphi_0^{1,\downarrow} =0,$ 
these processes satisfy
\begin{equation}\label{148}
d\varphi_t^{0,\uparrow} \le(1-\la)S_t d\varphi_t^{1,\downarrow}, \quad d\varphi_t^{0,\downarrow} \geq S_t d\varphi_t^{1,\uparrow}, \quad 0\le t\le T.
\end{equation}
\vskip6pt
The trading strategy $\varphi =(\varphi^0,\varphi^1)$ is called \textnormal{admissible} if there is $M>0$ such that the liquidation value $V_t^{liq}$ satisfies
\begin{equation}\label{149}
V_\tau^{liq}(\varphi^0,\varphi^1):=\varphi^0_\tau+(\varphi_\tau^1)^+(1-\la)S_\tau-(\varphi_\tau^1)^-S_\tau\geq -M,
\end{equation}
a.s., for all $[0,T]$-valued stopping times $\tau$.
\end{definition}

The processes $\varphi^0_t$ and $\varphi^1_t$ model the holdings at time $t$ in units of bond and stock respectively. We normalize the bond price by $B_t \equiv 1$. The differential notation in \eqref{148} needs some explanation. If $\varphi$ is continuous, then \eqref{148} has to be understood as the integral requirement.

\begin{equation}\label{p4}
\int^{\tau}_{\sigma} ((1-\lambda) S_t d\varphi^{1,\downarrow}_t - d\varphi^{0,\uparrow}_t) \geq 0, \qquad \qquad a.s.
\end{equation} 
for all stopping times $0 \leq \sigma \leq \tau \leq T$, and analogously for the second differential inequality in \eqref{148}.  The above integral makes pathwise sense as Riemann-Stieltjes intregral, as $\varphi$ is continuous and of finite variation and $S$ is c\`adl\`ag. Things become more delicate when we also consider jumps of $\varphi$: note that, for every stopping time $\tau$ the left and right limits $\varphi_{\tau_-}$ and $\varphi_{\tau_+}$ exist as $\varphi$ is of bounded variation. But the three values $\varphi_{\tau_-}, \varphi_{\tau}$ and $\varphi_{\tau_+}$ may very well be different. As in \cite{CampScha} we denote the increments by

\begin{equation}\label{181}
\Delta\varphi_{\tau} = \varphi_{\tau} - \varphi_{\tau_-}, \qquad \qquad \Delta_+\varphi_{\tau}=\varphi_{\tau_+}-\varphi_{\tau}.
\end{equation} 
For totally inaccessible stopping times $\tau$, the predictability of $\varphi$ implies that $\Delta\varphi_{\tau} = 0$ almost surely, while for accessible stopping times $\tau$ it may happen that $\Delta\varphi_{\tau} \neq 0$ as well as $\Delta_+\varphi_{\tau} \neq 0.$

To the assumption that \eqref{p4} has to hold true for the continuous part of $\varphi$ the following requirements therefore have to be added to take care of the jumps of $\varphi.$
\begin{equation}\label{177}
\Delta\varphi_{\tau}^{0,\uparrow} \leq (1-\lambda) S_{\tau_-} \Delta\varphi^{1,\downarrow}_{\tau}, \qquad \qquad \Delta\varphi_{\tau}^{0,\downarrow} \geq S_{\tau_-} \Delta\varphi^{1,\uparrow}_{\tau}
\end{equation} 
and in the case of right jumps
\begin{equation}\label{178}
\Delta_+\varphi_{\tau}^{0,\uparrow} \leq (1-\lambda) S_{\tau} \Delta_+\varphi^{1,\downarrow}_{\tau}, \qquad \qquad \Delta_+\varphi_{\tau}^{0,\downarrow} \geq S_{\tau} \Delta_+\varphi^{1,\uparrow}_{\tau},
\end{equation} 
holding true a.s. for all $[0,T]$-valued stopping times $\tau.$
Let us give an economic interpretation of the significance of \eqref{177} and \eqref{178}. For simplicity we let $\lambda =0.$ Think of a predictable time $\tau$, say the time $\tau$ of a speech of the chairman of the Fed. The speech does not come as a surprise. It was announced some time before which - mathematically speaking - corresponds to the predictability of $\tau$. 
It is to be expected that this speech will have a sudden effect on the price of a stock $S$, say a possible jump from $S_{\tau_-}(\omega) = 100$ to $S_{\tau}(\omega)= 110$ (recall that $S$ is assumed to be c\`adl\`ag). A trader may want to follow the following strategy: she holds a position of $\varphi^1_{\tau_-}(\omega)$ stocks until ``immediately before the speech". 
Then, one second before the speech starts, she changes the position from $\varphi^1_{\tau_-} (\omega)$ to $\varphi^1_{\tau} (\omega)$ causing an increment of $\Delta\varphi^1_{\tau}(\omega)$. 
Of course, the price $S_{\tau_-}(\omega)$ still applies, corresponding to \eqref{177}. Subsequently, the speech starts and the jump $\Delta S_{\tau} (\omega) = S_{\tau}(\omega) - S_{\tau_-} (\omega)$ is revealed. The agent may now decide ``immediately after learning the size of $\Delta S_{\tau} (\omega)$" to change her position from $\varphi^1_{\tau} (\omega)$ to $\varphi^1_{\tau_+} (\omega)$ on the base of the price $S_{\tau}(\omega)$ which corresponds to \eqref{178}.

\vskip10pt

We have chosen to define the trading strategy $\varphi$ by explicitly specifying both accounts, the holdings in bond $\varphi^0$ as well as the holdings in stock $\varphi^1.$ 
It would be sufficient to only specify $\varphi^1$ similarly as in the frictionless theory where we usually only specify the process $H$ in \eqref{1} which corresponds to $\varphi^1$ in the present notation. Given a predictable finite variation process $\varphi^1 =(\varphi^1_t)_{0\le t\le T}$ starting at
$\varphi^1_0 =0,$ which we canonically decompose into the difference $\varphi^1 =\varphi^{1,\uparrow} -\varphi^{1,\downarrow},$ we may \textnormal{define} the process $\varphi^0$ by
\begin{align*}
d\varphi^0_t =(1-\la) S_t d\varphi_t^{1,\downarrow} -S_t d\varphi^{1,\uparrow}_t.
\end{align*}
The resulting pair $(\varphi^0,\varphi^1)$ obviously satisfies \eqref{148} with equality holding true rather than inequality. Not withstanding, it is convenient in \eqref{148} to consider trading
strategies $(\varphi^0,\varphi^1)$ which allow for an inequality in \eqref{148}, $\mbox{i.e.}$ for ``throwing away money''. But it is clear from the preceding argument that we may always pass to a 
dominating pair $(\varphi^0,\varphi^1)$ where equality holds true in \eqref{148}.
\vskip10pt

In the theory of financial markets under transaction costs the super-martingale property of the value process is formulated in Proposition 1.6 below. First we have to recall a definition from \cite{DM82} which extends the notion of a super-martingale beyond the framework of c\`adl\`ag processes.

\begin{definition}
An optional process $X=(X_t) _{0 \leq t \leq T}$ is called an \textnormal{optional strong super-martingale} if, for all stopping times $0\leq \sigma \leq \tau \leq T$ we have 
\begin{equation}\label{179}
\mathbb{E} [X_\tau \mid \mathcal{F}_{\sigma}] \leq X_\sigma,
\end{equation}
where we impose that $X_{\tau}$ is integrable.
\end{definition}

An optional strong super-martingale can be decomposed in the style of Doob-Meyer which is known under the name of Mertens decomposition (see \cite{DM82}). $X$ is an optional strong super-martingale if and only if it can be decomposed into 
\begin{equation}\label{184}
X=M-A,
\end{equation}
where $M$ is a local martingale (and therefore c\`adl\`ag) as well as a super-martingale, and $A$ an increasing predictable process (which is l\`adl\`ag but has no reason to be c\`agl\`ad or c\`adl\`ag). This decomposition then is unique. 

One may also define the notion of a local optional strong supermartingale in an obvious way. In this case the process $M$ in \eqref{183} only is required to be a local martingale and not necessarily a super-martingale, while the requirements on $A$ remain unchanged.

\begin{proposition}\label{150}
Fix $S$, transaction costs $1>\la>0,$ and an admissible self-financing trading strategy $\varphi=(\varphi^0, \varphi^1)$ as above. Suppose that $(\widetilde{S},Q)$ is a consistent price system under transaction costs $\la$. Then the process
$$\widetilde{V}_t:= \varphi^0_t + \varphi^1_t \widetilde{S}_t, \qquad \qquad \qquad 0\leq t \leq T,$$
satisfies $\widetilde{V} \geq V^{\mbox{liq}}$ almost surely and is an optional strong super-martingale under $Q$.
\end{proposition}

\begin{proof}
The assertion $\widetilde{V} \geq V^{\mbox{liq}}$ is an obvious consequence of $\widetilde{S} \in [(1-\lambda) S,S]$.

We have to show that $\widetilde{V}$ decomposes as in \eqref{184}. Arguing formally, we may apply the product rule to obtain
\begin{equation}\label{182}
d\widetilde{V}_t = (d\varphi^0_t + \widetilde{S}_t d\varphi^1_t) + \varphi^1_t d\widetilde{S}_t
\end{equation}
so that
\begin{equation}\label{183}
\widetilde{V}_t=\int^t_0(d\varphi^0_u + \widetilde{S}_u d\varphi^1_u) + \int^t_0\varphi^1_u d\widetilde{S}_u.
\end{equation}
The first term in \eqref{183} is decreasing by \eqref{148} and the fact that $\widetilde{S}\in[(1-\la)S,S]$. The second term defines, at least formally speaking, a local Q-martingale as $\widetilde{S}$ is so. Hence the sum of the two integrals should be an (optional strong) super-martingale. 
\vskip10pt
The justification of the above formal reasoning deserves some care (compare the proof of Lemma 8, in \cite{CampScha}). Suppose first that $\varphi$ is continuous. In this case $\varphi$ is a semi-martingale so that we are allowed to apply Itô calculus to $\widetilde{V}$. Formula \eqref{183} therefore makes perfect sense as an Itô integral, bearing in mind that $\varphi$ has finite variation, which coincides with the pointwise interpretation of the integral via partial integration. 
The first integral in \eqref{183} is a well-defined decreasing predictable process. As regards the second integral, note that by the admissibility of $\varphi$ it is uniformly bounded from below. Hence by a result of Ansel-Stricker (\cite{AS94}, see also \cite{St02}) it is a local Q-martingale as well as a super-martingale.  Hence $\widetilde{V}$ is indeed a super-martingale under $Q$ (in the classical c\`adl\`ag sense).
\vskip10pt
Passing to the case when $\varphi$ is allowed to have jumps, the process $\widetilde{V}$ need not be c\`adl\`ag anymore. It still is an optional process and we have to verify that it decomposes as in \eqref{184}. Assume first that $\varphi$ is of the form 
\begin{equation}\label{185}
\varphi_t = (f^0, f^1) \mathbbm{1}_{\rrbracket\tau,T \rrbracket} (t),
\end{equation}
where $(f^0, f^1)=\Delta_+(\varphi^0_{\tau}, \varphi^1_{\tau})$ are $\mathcal{F}_{\tau}$-measurable bounded random variables verifying \eqref{178} and $\tau$ is a $[0,T]$-stopping time. 
We obtain 

\begin{align}\label{185a}
\widetilde{V}_t &= [\Delta_+\varphi^0_{\tau} +(\Delta_+\varphi^1_{\tau})\widetilde{S}_t] \mathbbm{1}_{\rrbracket \tau,T \rrbracket}{(t)}\nonumber\\
&=[\Delta_+\varphi^0_{\tau} +(\Delta_+\varphi^1_{\tau})\widetilde{S}_{\tau}] \mathbbm{1}_{\rrbracket \tau,T \rrbracket}{(t)} + (\Delta_+\varphi^1_{\tau}) (\widetilde{S}_t - \widetilde{S}_{\tau}) \mathbbm{1}_{\rrbracket \tau,T \rrbracket}{(t)}.
\end{align}
Again, the first term is a decreasing predictable process and the second term is a local martingale under $Q$.
\vskip10pt 
Next assume that $\varphi$ is of the form
\begin{equation}\label{186}
\varphi_t = (f^0, f^1) \mathbbm{1}_{\llbracket \tau,T \rrbracket}(t),
\end{equation}
where $\tau$ is a predictable stopping time, and $(f^0, f^1)=\Delta (\varphi^0_{\tau}, \varphi^1_{\tau})$ are bounded $\mathcal{F}_{\tau_-}$-measurable random variables verifying \eqref{177}.
Similarly as in \eqref{185a} we obtain 

\begin{align}\label{186a}
\widetilde{V}_t &= [\Delta\varphi^0_{\tau} +(\Delta\varphi^1_{\tau})\widetilde{S}_t] \mathbbm{1}_{\llbracket \tau,T \rrbracket}{(t)}\nonumber\\
&=[\Delta\varphi^0_{\tau} +(\Delta\varphi^1_{\tau})\widetilde{S}_{\tau_-}] \mathbbm{1}_{\llbracket \tau,T \rrbracket}{(t)} + (\Delta\varphi^1_{\tau}) (\widetilde{S}_t - \widetilde{S}_{\tau_-}) \mathbbm{1}_{\llbracket \tau,T \rrbracket}{(t)}.
\end{align}

Once more, the first term is a decreasing predictable process (this time it is even c\`adl\`ag) and the second term is a local martingale under $Q$.
\vskip10pt
Finally we have to deal with a general admissible self-financing trading strategy $\varphi$. To show that $\widetilde{V}$ is of the form \eqref{184} we first assume that the total variation of $\varphi$ is uniformly bounded. We decompose $\varphi$ into its continuous and purely discontinuous part $\varphi = \varphi^c + \varphi^{pd}$. We also may find a sequence $(\tau_n)^{\infty}_{n=1}$ of $[0,T]\cup\{\infty\}$-valued stopping times such that the supports $(\llbracket\tau_n \rrbracket)^\infty_{n=1}$ are mutually disjoint and $\bigcup\limits^\i_{n=1} \llbracket\tau_n \rrbracket$ exhausts the right jumps of $\varphi$. Similarly, we may find a sequence $(\tau_n^p)^{\infty}_{n=1}$ of predictable stopping times such that their supports $(\llbracket\tau_n^p\rrbracket)^\infty_{n=1}$ are mutually disjoint and $\bigcup\limits^\i_{n=1} \llbracket\tau_n^p\rrbracket$ exhausts the left jumps of $\varphi$.
We apply the above argument to $\varphi^c$, and to each $(\tau_n, \Delta_+\varphi_{\tau_n})$ and $(\tau_n^p, \Delta \varphi_{\tau_n^p})$, and sum up the corresponding terms in \eqref{183}, \eqref{185a} and \eqref{186a}. This sum converges to $\widetilde{V} = M-A$, where $M$ is a local Q-martingale and $A$ an increasing process, as we have assumed that the total variation of $\varphi$ is bounded (compare \cite{KabaStri} and the proof of Lemma 8 in \cite{CampScha}). By the boundedness from below we conclude that $M$ is also a super-martingale.
\vskip10pt
Passing to the case where $\varphi$ has only finite instead of uniformly bounded variation, we use the predictability of $\varphi$ to find a localizing sequence $(\sigma_k)^\infty_{k=1}$ such that each stopped process $\varphi^{\sigma_k}$ has uniformly bounded variation. Apply the above argument to each $\varphi^{\sigma_k}$ to obtain the same conclusion for $\varphi$.

Summing up, we have shown that $\widetilde{V}$ admits a Mertens decomposition \eqref{184} and therefore is an optional strong super-martingale.
\end{proof}

We can now state the analogous result to Proposition \ref{pro4.12} in the presence of transaction costs.

\begin{theorem}\label{pro4.11}
Fix the c\`adl\`ag, adapted process $S$ and $1>\la >0$ as above, and suppose that $S$ satisfies $(CPS^{\la'})$, for \textnormal{each} $1 > \la'>0.$

Let $\varphi =(\varphi^0_t,\varphi^1_t)_{0\le t\le T}$ be an admissible, self-financing trading strategy under transaction costs $\la,$ starting with zero endowment, and suppose that there is 
$x>0$ s.t. for the terminal liquidation value $V_T^{liq}$ we have a.s.
\begin{equation}\label{169}
V_T^{liq}(\varphi^0,\varphi^1) =\varphi^0_T+({\varphi_T^1})^+(1-\la)S_T-({\varphi^1_T})^- S_T \geq -x.
\end{equation}
We then also have that
\begin{equation}\label{170}
V_{\tau}^{liq}(\varphi^0,\varphi^1)=\varphi^0_\tau + ({\varphi^1_\tau})^+(1-\la)S_\tau-({\varphi^1_\tau})^- S_\tau \geq -x,
\end{equation}
a.s., for every stopping time $0\le \tau\le T.$
\end{theorem}

\begin{proof}
Supposing that \eqref{170} fails, we may find $\tfrac{\la}{2} >\alpha >0,$ and a stopping time $0\le \tau\le T,$ such that either $A=A_+$ or $A=A_-$ satisfies $\p[A]>0,$ where
\begin{align}
A_+=\{\varphi^1_\tau \geq 0, \ &\varphi^0_\tau+\varphi^1_\tau \tfrac{1-\la}{1-\alpha} S_\tau <-x\}, \label{171} \\
A_-=\{\varphi^1_\tau \le 0, \ &\varphi^0_\tau+\varphi^1_\tau (1-\alpha)^2 S_\tau < -x\}. \label{172}
\end{align}
Indeed, focusing on \eqref{171} and denoting by $A_+(\alpha)$ the set in \eqref{171} we have $\cup_{\alpha>0} A_+(\alpha)=\{\varphi^1_\tau \geq 0, \varphi^0_\tau + \varphi^1_\tau (1-\lambda)S_\tau < -x\}$, showing that the failure of \eqref{170} implies the existence of $\alpha >0$ such that $\mathbb{P}[A] > 0.$
 
Choose $0 <\la' <\alpha$ and a $\la'$-consistent price system $(\widetilde{S},Q).$ As $\widetilde{S}$ takes values in $[(1-\la')S,S]$, we have that $(1-\alpha)\widetilde{S}$ as well as
$\tfrac{1-\la}{1-\alpha}\widetilde{S}$ take values in $[(1-\la)S,S]$ as $(1-\lambda')(1-\lambda)>(1-\lambda)$ and $(1-\lambda')\frac{1-\lambda}{1-\alpha} > 1-\lambda.$ It follows that $((1-\alpha)\widetilde{S},Q)$ as well as $(\tfrac{1-\la}{1-\alpha}\widetilde{S},Q)$ are consistent price systems
under transaction costs $\la.$ By Proposition 1.6 we obtain that
$$\Big(\varphi^0_t+\varphi^1_t(1-\alpha)\widetilde{S}_t\Big)_{0\le t\le T} ~ \mbox{and} \ \Big(\varphi^0_t+\varphi^1_t\tfrac{1-\la}{1-\alpha} \widetilde{S}_t\Big)_{0\le t\le T}$$
are optional strong $Q$-super-martingales. Arguing with the second process using $\widetilde{S} \leq S,$ we obtain from \eqref{171} the inequality
\begin{align*}
\E_Q[V^{liq}_T \mid A_+] &\leq \E_Q\left[\varphi^0_T+\varphi^1_T\frac{1-\la}{1-\alpha}\widetilde{S}_T \Big|A_+\right]\\
& \le \E_Q\left[\varphi^0_\tau+\varphi^1_\tau\frac{1-\la}{1-\alpha} \widetilde{S}_\tau \Big|A_+\right]\\
& \le \E_Q\left[\varphi^0_\tau+\varphi^1_\tau\frac{1-\la}{1-\alpha} S_\tau \big|A_+\right]<-x.
\end{align*}
Arguing with the first process and using that $\widetilde{S}\geq (1-\la') S\geq(1-\al)S$ (which implies that $\varphi^1_\tau (1-\al)\widetilde{S}_\tau \le \varphi^1_\tau(1-\al)^2S_\tau$ 
on $A_-$) we obtain from \eqref{172} the inequality
\begin{align*}
\E_Q[V^{liq}_T \mid A_-] &\leq  \E_Q\left[\varphi^0_T+\varphi^1_T(1-\al)\widetilde{S}_T |A_-\right]\\
& \le \E_Q\left[\varphi^0_\tau+\varphi^1_\tau(1-\al)\widetilde{S}_\tau |A_-\right]\\
& \le \E_Q\left[\varphi^0_\tau+\varphi^1_\tau(1-\al)^2 S_\tau |A_-\right]<-x.
\end{align*}
Either $A_+$ or $A_-$ has strictly positive probability; hence we arrive at a contradiction to $\Vliq_T \geq - x$ almost surely.
\end{proof}

\section{The numéraire-free setting}

In this section we derive results analoguous to Proposition \ref{150} and Theorem \ref{pro4.11} in a numéraire-free setting. This is inspired by the discussion of the numéraire-based versus numéraire-free setting in \cite{GRS08} and \cite{S14} (compare also \cite{DS95}, \cite{KS02}, \cite{Y98}, \cite{Y05}). 

We complement the above notions of admissibility and consistent price systems by the following numéraire-free variants.

\begin{definition}
In the setting of Definition \ref{def4.2} we call a self-financing strategy $\varphi$ {\it admissible in a numéraire-free sense} if there is $M>0$ such that
\begin{equation}\label{N1}
V_\tau^{liq}(\varphi^0,\varphi^1):=\varphi^0_\tau+(\varphi_\tau^1)^+(1-\la)S_\tau-(\varphi_\tau^1)^-S_\tau\geq -M (1+S_\tau), \quad \quad \mbox{a.s.},
\end{equation}
for each $[0,T]$-valued stopping time $\tau$.
\end{definition}

While the control of the portfolio process $\varphi$ in \eqref{149} is in terms of $M$ units of bond (which is considered as numéraire), the present condition \eqref{N1} stipulates that the risk involved by the trading strategy $\varphi$ can be super-hedged by holding $M$ units of bond plus $\frac{M}{1-\lambda}$ units of stock.

\begin{definition}\label{def.N2}
Fix $1 > \lambda \geq 0.$ In the setting of Definition \ref{def4.1} we call a pair $(\widetilde{S},Q)=((\widetilde{S}_t)_{0 \leq t \leq T,Q})$ satisfying \eqref{p3} a consistent price process {\it in the non-local sense} if $\widetilde{S}$ is a true martingale under $Q$, {\it not only a local martingale}.
\end{definition}

The passage from the numéraire-based to numéraire-free admissibility for the primal objects, i.e.~the trading strategies $\varphi$, perfectly corresponds to the passage from local martingales to martingales in Definition \ref{def.N2} for the dual objects, i.e.~the consistent price systems. This is the message of the two subsequent results (compare also \cite{S14}).

\begin{proposition}\label{prop2.3}
In the setting of Proposition \ref{150}  fix a self-financing trading strategy $\varphi=(\varphi^0, \varphi^1)$ which we now assume to be {\it admissible in the numéraire-free sense}. Also fix $(\widetilde{S}, Q)$ which we now assume to be a $\lambda$-consistent price system {\it in the non-local sense}, i.e.~$\widetilde{S}$ is a true $Q$-martingale.
We again may conclude that the process 
$$\widetilde{V}_t:= \varphi^0_t + \varphi^1_t \widetilde{S}_t, \qquad \qquad \qquad 0\leq t \leq T,$$
satisfies $\widetilde{V} \geq V^{liq}$ almost surely and is an optional strong super-martingale under $Q$.
\end{proposition}

\begin{proof}
We closely follow the proof of Proposition \ref{150} which carries over verbatim, also under the present weaker assumption of numéraire-free admissibility. Again, we conclude that the second integral in \eqref{183} is a local $Q$-martingale from the fact that $\widetilde{S}$ is a local $Q$-marginale and $\varphi^1$ is predictable and of finite variation. The only subtlety is the following:~contrary to the setting of Proposition \ref{150} we now may only deduce the obvious implication that $\widetilde{V}=(\widetilde{V}_t)_{0 \leq t \leq T}$ is a {\it local} optional strong super-martingale under $Q$. 

What needs extra work is an additional argument which finally shows that the word {\it local} may be dropped, i.e.~that $\widetilde{V}$ again is an optional strong super-martingale under $Q$.

By the numéraire-free admissibility condition we know that there is some $M>0$ such that, for all $[0,T]$-valued stopping times $\tau$,
\begin{equation}\label{N5}
\widetilde{V}_\tau \geq V_\tau^{liq} \geq -M(1+S_\tau), \qquad \mbox{a.s.}
\end{equation}
We also know that $\widetilde{S}$ is a uniformly integrable martingale under $Q$. Hence the family of random variables $\widetilde{S}_\tau$ as well as that of $S_\tau$ (note that $S_\tau \leq \frac{\widetilde{S}_\tau}{1-\lambda}$), where $\tau$ ranges through the $[0,T]$-valued stopping times, is uniformly integrable.

We have to show that, for all stopping times $0 \leq \rho \leq \sigma \leq T$ we have
\begin{equation}\label{N5a}
\mathbb{E}_Q[\widetilde{V}_\sigma|\mathcal{F}_\rho] \leq \widetilde{V}_\rho.
\end{equation}
We know that $\widetilde{V}$ is a local optional strong super-martingale under $Q$, so that there is a localizing sequence $(\tau_n)^\infty_{n=1}$ of stopping times such that
\begin{equation}\label{N5b}
\mathbb{E}_Q[\widetilde{V}_{\sigma \wedge \tau_n} | \mathcal{F}_{\rho \wedge \tau_n}] \leq \widetilde{V}_{\rho \wedge \tau_n}, \quad n \geq 1.
\end{equation}
Using \eqref{N5} we may deduce \eqref{N5a} from \eqref{N5b} by the (conditional version of the) following well-known variant of Fatou's lemma: Let $(f_n)^\infty_{n=1}$ be a sequence of random variables on $(\Omega, \mathcal{F}, \mathbb{R})$ converging almost surely to $f_0$ and such that the negative parts $(f_n^-)^\infty_{n=1}$ are uniformly $Q$-integrable. Then
$$\mathbb{E}_Q[f_0] \leq \liminf_{n \to \infty} \mathbb{E}_Q[f_n].$$ 
\end{proof}

\begin{remark}
We have assumed in Proposition \ref{150} as well as in the above Proposition \ref{prop2.3} that $Q$ is equivalent to $\mathbb{P}.$ In fact, we may also assume that $Z^0_T$ vanishes on a non-trivial set so that $Q$ is only absolutely continuous w.r.~to $\mathbb{P}$. The assertions of the two propositions still remain valid for $\mathbb{P}$-absolutely continuous $Q$, provided that we replace the requirements {\it almost surely} by {\it $Q$-almost surely}.
\end{remark}
We now state and prove the numéraire-free version of Theorem \ref{pro4.11}.

\begin{theorem}
In the setting of Theorem \ref{pro4.11} suppose now that $S$ satisfies $(CPS^{\lambda'})$ in the non-local sense, for each $1 > \lambda' > 0.$
As in Theorem \ref{pro4.11}, let $\varphi$ be admissible, but now in the numéraire-free sense, and let $x > 0$ such that
\begin{equation}\label{t1}
V_T^{liq}(\varphi^0, \varphi^1) \geq -x.
\end{equation}
We then also have
\begin{equation}\label{t2}
V_\tau^{liq}(\varphi^0, \varphi^1) \geq -x,
\end{equation}
a.s., for every stopping time $0 \leq \tau \leq T.$
\end{theorem}

\begin{proof}
The proof of Theorem \ref{pro4.11} carries over verbatim to the present setting, replacing the application of Proposition \ref{150} by an application of its numéraire-free version Proposition \ref{prop2.3}.
\end{proof}

\section{A Counter-Example}

The assumption $(CPS^{\la'})$, for {\it each} $\la'>0,$ cannot be dropped in Proposition \ref{pro4.11} as shown by the example presented in the next lemma.
\begin{lemma}\label{l4.12}
Fix $1 >\la \geq \la' >0$ and $C>1$. There is a continuous process $S=(S_t)_{0\le t \le 1}$ satisfying $(CPS^{\la'}),$ and a $\la$-self-financing, admissible trading strategy 
$(\varphi^0,\varphi^1)=(\varphi^0_t,\varphi^1_t)_{0\le t\le 1}$ such that
\begin{equation}\label{173}
V^{liq}_1(\varphi^0,\varphi^1)  \geq -1, \qquad\qquad\qquad \mbox{a.s.}
\end{equation}
while
\begin{equation}\label{174}
\hspace*{-2.7cm}\p\left[ V^{liq}_{\tfrac12} (\varphi^0,\varphi^1) \le -C\right] >0.
\end{equation}
\end{lemma}

\begin{proof}
In order to focus on the central (and easy) idea of the construction we first show the assertion for the constant $C=2-\la$ and under the assumption $\la = \la'$. In this case we can give a deterministic example, 
$\mbox{i.e.}~S,\varphi^0$ and $\varphi^1$ will not depend on the random element $\omega \in\Omega.$

Define $S_0=S_1=1,$ and $S_{\tfrac12} =1-\la$ where we fix $T=1$.

To make $S=(S_t)_{0\le t\le T}$ continuous, we interpolate linearly, i.e.
\begin{equation}\label{175}
S_t=1-2t\la, \qquad \qquad 0\le t\le \tfrac12, 
\end{equation}
\begin{equation}\label{176}
S_t=1-2(1-t)\la, \qquad \tfrac12 \le t\le 1.
\end{equation}

Note that condition $(CPS^\la)$ is satisfied, as the constant process $\widetilde{S}_t\equiv(1-\la)$ defines a 
$\la$-consistent price system: it trivially is a martingale (under any probability measure) and takes values in $[(1-\la)S,S].$

Starting from the initial endowment $(\varphi^0_{0},\varphi^1_{0})=(0,0),$ we might invest, at time $t=0,$ the maximal amount into the stock so that at time $t=1$ condition
\eqref{173} holds true. In other words, we let $\varphi^1_{0_+} =-\varphi^0_{0_+}$ be the biggest number such that
\begin{align*}
(1-\la)\varphi^1_{0_+} +\varphi^0_{0_+} \geq -1,
\end{align*}
which clearly gives $\varphi^1_{0_+}=\tfrac{1}{\la}.$ Hence $(\varphi^0_t,\varphi^1_t)=(-\tfrac{1}{\la},\tfrac{1}{\la}),$ for all $0 < t\le T,$ is a self-financing strategy, starting at
$(\varphi^0_0,\varphi^1_0)=(0,0)$ for which \eqref{173} is satisfied.
Looking at \eqref{174} we calculate
\begin{align*}
V_{\tfrac12}(\varphi^0,\varphi^1)=(1-\la)\cdot(1-\la)\cdot\tfrac{1}{\la}-\tfrac{1}{\la}=-2+\la.
\end{align*}

\bigskip

In order to replace $\la'= \la$ by an arbitrarily small constant $\la' > 0$, and $C=2-\la$ by an arbitrarily large constant $C>1,$ we make the following observation: if the initial endowment $(\varphi^0_{0},\varphi^1_{0})=(0,0)$
were replaced by $(\varphi^0_{0},\varphi^1_{0})=(M,0),$ for some large $M$, the agent could play the above game on a larger scale: she could choose $(\varphi^0_t,\varphi^1_t)=
(M-\tfrac{M+1}{\la},\tfrac{M+1}{\la}),$ for $0< t\le 1$, to still satisfy \eqref{173}:
\begin{align*}
V_1(\varphi^0,\varphi^1)=M-\tfrac{M+1}{\la} +(1-\la)\tfrac{M+1}{\la} =-1.
\end{align*}
As regards the liquidation value $\Vliq_{\tfrac12}$, we now assume $S_{\tfrac12} = 1-\la'$ (instead of $S_{\tfrac12} = 1-\la$ in \eqref{175} and \eqref{176}) to make sure that $(CPS^{\la'})$ holds true. The liquidation value at time $t=\frac{1}{2}$ then becomes

\begin{align*}
\Vliq_{\tfrac12}(\varphi^0,\varphi^1)&=M-\tfrac{M+1}{\la}+(1-\la)(1-\la')\tfrac{M+1}{\la}\\
&=M-(M+1) [1+\la'(\frac{1}{\la}-1)]
\end{align*}
which tends to $-\infty$, as $M\to \infty$ in view of $0<\la'\leq \la<1$.

Turning back to the original endowment $(\varphi^0_0,\varphi^1_0)=(0,0),$ the idea is that, during the time interval $[0,\tfrac14],$ the price process $S$ provides the agent 
with the opportunity to become rich with positive probability, i.e. \linebreak $\p[(\varphi^0_{\tfrac14},\varphi^1_{\tfrac14})=(M,0)] >0.$
We then play the above game, conditionally on the event $\{(\varphi^0_{\tfrac14},\varphi^1_{\tfrac14})=(M,0)\}$ and with $[0,1]$ replaced by $[\tfrac14,1].$

The subsequent construction makes this idea concrete. Let $(\F_t)_{0\le t\le 1}$ be generated by a Brownian motion $(W_t)_{0\le t\le 1}.$
Fix disjoint sets $A_+$ and $A_-$ in $\F_{\tfrac18}$ such that $\p[A_+]=\tfrac{1}{2\widetilde{M}-1}$ and $\p[A_-]=1-\p[A_+],$ where $\widetilde{M}>1$ is defined by $M=-1+\widetilde{M}(1-\la').$
The set $A_+$ is split into two sets $A_{++}$ and $A_{+-}$ such that $A_{++}$ and $A_{+-}$ are in $\F_{\tfrac14}$ and
\begin{align*}
\p\left[A_{++} \bigg|\F_{\tfrac18}\right]=\p\left[A_{+-} \bigg|\F_{\tfrac18}\right]=\frac12\mathbbm{1}_{A_+}.
\end{align*}

We define $S_{\tfrac14}$ by
\begin{align*}
S_{\tfrac14} = \begin{cases}
2\widetilde{M}-1 &\qquad \mbox{on} ~ A_{++}\\
1 &\qquad \mbox{on} ~ A_{+-} \\
\frac12 &\qquad \mbox{on} ~ A_{-}
\end{cases}
\end{align*}
and
\begin{equation}
S_t=\E\left[S_{\tfrac14} |\F_t\right], \qquad\qquad 0\le t\le \frac14,
\end{equation}
so that $(S_t)_{0\le t\le \frac14}$ is a continuous $\p$-martingale. The numbers above were designed in such a way that

\begin{align*}
\hspace*{-1.4cm}S_0&=1,
\end{align*}
and
\begin{align*}
S_{\tfrac18} &= \begin{cases}
\widetilde{M} &\quad\mbox{on} ~ A_{+}\\
\frac12 &\quad\mbox{on} ~ A_{-}
\end{cases}
\end{align*}

To define $S_t$ also for $\tfrac14 <t\le 1$ we simply let $S_t=S_{\tfrac14}$ on $A_{++} \cup A_-$ while, conditionally on $A_{+-}$, we repeat the above deterministic 
construction on $[\tfrac14 ,1]:$
\begin{align*}
S_t&=1-4(t-\tfrac14)\la', \qquad &\tfrac14\le t\le \tfrac12,\\
S_t&=1-2(1-t)\la', \qquad &\tfrac12\le t\le 1.
\end{align*}

This defines the process $S.$ Condition $(CPS^{\la'})$ is satisfied as $(\widetilde{S}_t)_{0\le t\le 1}: =((1-\la')S_{t\wedge\tfrac14})_{0\le t\le 1}$ is a $\p$-martingale taking values in the
bid-ask spread $[(1-\la')S_t,S_t]_{0\le t\le 1}.$

Let us now define the strategy $(\varphi^0,\varphi^1):$ starting with $(\varphi^0_0,\varphi^1_0)=(0,0)$ we define $(\varphi^0_t,\varphi^1_t)=(-1,1),$ for 
$0< t\leq\tfrac18.$ In prose: the agent buys one stock at time $t=0$ and holds it until time $t=\tfrac18.$ At time $t=\tfrac18$ she sells the stock again, so that
$(\varphi^0_{\tfrac18},\varphi^1_{\tfrac18})=(-1+\tfrac{(1-\la)}{2},0)$ on $A_-,$ while $(\varphi^0_{\tfrac18_+},\varphi^1_{\tfrac18_+})=(-1+\widetilde{M}(1-\la'),0)=(M,0)$ on $A_+.$

On $A_-$ we simply define $(\varphi^0_t,\varphi^1_t)=(-1+\tfrac{1-\la}{2},0),$ for all $\tfrac18 < t\le 1$ and note that \eqref{173} is satisfied on $A_-$.

On $A_+$ we define $(\varphi^0_t,\varphi^1_t)=(M,0),$ for $\tfrac18< t\leq\tfrac14.$ In prose: during $]\tfrac18,\tfrac14]$ the agent does not invest into the stock and is happy about the $M$ bonds in her portfolio. At time $t=\tfrac14$ we distinguish two cases: on $A_{++}$ we continue to define 
$(\varphi^0_t,\varphi^1_t)=(M,0),$ also for $\tfrac14< t\le 1.$ On $A_{+-}$ we let $(\varphi^0_t,\varphi^1_t)=(M-\tfrac{M+1}{\la},\tfrac{M+1}{\la}),$ for 
$\tfrac14< t\le 1.$ As discussed above, inequality \eqref{173} then holds true almost surely, while $V_{\tfrac12}(\varphi^0,\varphi^1)$ attains the value $M-(M+1)[1+\la'(\frac{1}{\la}-1)]$ which tends to $-\infty$ as $M$ tends to $\infty$. This happens with positive probability $\p[A_{+-}]>0.$

The construction of the example now is complete.
\end{proof}

\begin{acka}
I warmly thank Irene Klein without whose encouragement this note would not have been written and who strongly contributed to its shaping. Thanks go also to Christoph Czichowsky for his advice on some of the subtle technicalities of this note. I thank an anonymous referee for careful reading and for pointing out a number of inaccuracies.
\end{acka}

\bibliography{appendix.bib}

\end{document}